\documentclass[10pt]{article}

\usepackage[accepted]{tmlr}

\usepackage{microtype}
\usepackage{graphicx}
\usepackage{subcaption}
\usepackage{booktabs}
\usepackage{csquotes}
\usepackage{physics}
\usepackage{mathrsfs}
\usepackage{dsfont}
\usepackage{tikz}
\usepackage{wrapfig}
\usetikzlibrary{arrows.meta, positioning, decorations.markings}
\usepackage{amsmath}
\usepackage{amssymb}
\usepackage{mathtools}
\usepackage{amsthm}
\usepackage{hyperref}
\usepackage{url}
\usepackage[capitalize,noabbrev]{cleveref}

\theoremstyle{plain}

\newtheorem{proposition}{Proposition}

\theoremstyle{definition}

\theoremstyle{remark}

\newcommand{\anisotropicGaussianOne}{figures/gaussian_rmh_across}
\newcommand{\anisotropicGaussianTwo}{figures/gaussian_rmh_along}
\newcommand{\gaussianMultirun}{figures/multirun}

\newcommand{\rosenbrockNonzero}{figures/rosenbrock_hmc_nonzero}
\newcommand{\rosenbrockZero}{figures/rosenbrock_hmc_zero}
\newcommand{\snOneA}{figures/sn1a}
\newcommand{\multirun}{figures/multirun}

\title{Thermodynamic Cyclic Processes\\with Markov Samplers in Bayesian Inference}

\author{\name Heinrich von Campe \email heinrich.campe@iwr.uni-heidelberg.de \\
      \addr Interdisziplin{\"a}res Zentrum f{\"u}r wissenschaftliches Rechnen, Universit{\"a}t Heidelberg\\
      \addr Tübingen AI Center, University of Tübingen \\
      \addr ELLIS Institute Tübingen \\
          \addr Zuse School ELIZA
      \AND
      \name Bj{\"o}rn Malte Sch{\"a}fer \email bjoern.malte.schaefer@uni-heidelberg.de \\
      \addr Interdisziplin{\"a}res Zentrum f{\"u}r wissenschaftliches Rechnen, Universit{\"a}t Heidelberg\\
      \addr Zentrum f{\"u}r Astronomie der Universit{\"a}t Heidelberg, Astronomisches Rechen-Institut}

\def\month{08}
\def\year{2026}
\def\openreview{\url{https://openreview.net/forum?id=88KWqihymD}}

\begin{document}

\maketitle

\begin{abstract}
The concept of Markov chain Monte Carlo (MCMC) cycles, an analogy to cyclic
processes in heat engines, is presented in order to examine Bayesian inference
problems. In this effort, we develop adaptive ensemble schedulers that allow
the tuning of external parameters of a Bayesian canonical ensemble during an MCMC
run, realising the MCMC cycles in practice. We run these cycles on different
statistical models. As a fundamental insight, we find (both theoretically and in
practice) that such systems can produce a non-zero net work output if and only
if the considered model is non-Gaussian. As such, they may serve as a measure of
non-Gaussianity in Bayesian inference, which we test on an example from
supernova cosmology.
\end{abstract}

\section{Introduction}
One of the prime motivations for the development of the theory of thermodynamics
in the 19th century was the study of heat engines as the first apparatuses that
could translate heat into work as a usable form of energy
\citep{history_of_thermo}. In the beginning, the nature of gases, their working
substances, was poorly understood on the microscopic level. Thus, thermodynamics
was developed as an \emph{effective} theory to accurately capture their
macroscopic behaviour nonetheless. Subsequently, thermodynamics became a
powerful theory that could be applied to systems of few as well as many degrees
of freedom. When the \enquote{molecular picture} of matter was established,
statistical physics was developed by Boltzmann and others to bridge the gap to
the macroscopic thermodynamic notions.  Much later,
\citet{jaynes_information_1957} found that thermodynamics and statistical
physics may indeed also be understood as a theory of information.

Conversely, Bayesian inference and Markov Chain Monte-Carlo (MCMC) sampling 
can be studied from the point of view of statistical physics and thermodynamics
as the following construction demonstrates. Bayes' theorem,
\begin{equation} \label{eq:bayes}
    p(\theta \mid y) = \frac{\mathcal{L}(y\mid\theta)\pi(\theta)}{p(y)}\,,
\end{equation}
with the likelihood $\mathcal{L}(y\mid\theta)$, parameter $\theta$, the data $y$,
and the prior $\pi(\theta)$, requires the evidence as a normalization constant
\begin{equation} \label{eq:evidence}
    p(y) = \int \mathrm{d}^n \theta \, \mathcal{L}(y\mid\theta)\pi(\theta)\,.
\end{equation}
Additional parameters, the temperature $T$ and the sources $J$, extend this
integral to a canonical partition sum in the sense of statistical
physics
\begin{equation} \label{eq:canZ}
    Z(T, J) = \int \mathrm{d}^n \theta \left(\mathcal{L}(y\mid\theta)\pi(\theta)\right)^{1/T} e^{J \cdot \theta / T}\,.
\end{equation}
To this one may then apply the tools from statistical mechanics and
thermodynamics to obtain a theory of the Bayesian canonical ensemble and its
extensions \citep{partitionfunction101, partition_functional, partition_info}.
The sources $J \in \mathds{R}^n$ are well known from statistical physics and
quantum field theory; they can be used to obtain the moments (cumulants) of the
posterior by evaluating the (logarithmic) derivative of the partition sum $Z$ at
$T=1$ and $J=0$.  In practice, the integrals often cannot be carried out
analytically, which is why one resorts to Markov chain Monte-Carlo (MCMC)
sampling algorithms such as the Rosenbluth-Metropolis-Hastings (RMH) algorithm
\citep{metropolis,hastings}.  Much like in a physical gas, such algorithms work
by simulating a random walk in the sample space. The samples of the posterior
are then given by the trajectory of the sampling \enquote{particles}. The
thermodynamic point of view has yielded new convergence criteria for MCMC
algorithms \citep{partition_thermal} as well as entirely new sampling algorithms
\citep{partition_macro, mams}.

Aside from sampling, a different use case for the RMH algorithm is called
\emph{simulated annealing} with the aim of finding the global maximum of the
posterior w.\,r.\,t.\@ the position $\theta$ \citep{simulated_annealing}. This
is accomplished by gradually decreasing the temperature $T$ as the sampler runs
such that the sampler's motion is more and more restricted around the maximum.

\paragraph{Contributions}
The main aim of this paper is the construction and study of an analogue to the
cyclic processes of mechanical heat engines in Bayesian inference. We first develop
the necessary mechanisms by generalising the idea of simulated annealing in
order to vary \emph{both} the temperature $T$ and the sources $J$ during a
simulation run.  With this, we establish a new method named MCMC cycles where
both parameters are changed in turn, forming a closed loop in $(T,J)$ space,
just like in heat engines.  We tune the parameters in a thermodynamically
consistent way by introducing \emph{adaptive ensemble schedulers} inspired by
the notion of quasi-staticity (see section \ref{sec:tunable_rmh}). The
thermodynamic theory derived from the partition sum \eqref{eq:canZ} allows us to
make predictions about the exchanged heat and performed work during the MCMC
cycles (section \ref{sec:thermodynamics}). Our further aims are to test whether
they work in practice, whether our theoretical predictions may be confirmed and
what properties these cycles exhibit (section \ref{sec:application}). As we will
see, the MCMC cycles function differently depending on the (non-)Gaussianity of
the sampled density, allowing us to use them as a diagnostic for that which we
test on an example of Bayesian inference in cosmology with type Ia supernovae.
Throughout the paper, we use index notation with the Einstein summation
convention that we briefly describe in Appendix~\ref{app:einstein_summation}.

\paragraph{Related work}
MCMC methods were originally developed for use in statistical physics
\citep{metropolis,hastings} but have since been applied to many other fields,
Lattice QCD \citep{mcmcLattice,hmc_original} and Bayesian inference
\citep{mcmcInference,mcmcImages} among them. The thermodynamic properties of
heat engines have been studied since the 19th century \citep{carnotHeatengines,
clausius_2nd_law, thomson_thermo}. The general thermodynamic nature of Bayesian
inference and MCMC methods has been studied extensively
\citep{jaynes_information_1957,partitionfunction101,partition_thermal,
partition_macro, partition_functional, partition_info}.  While the work by
\citet{cyclical_mcmc_sampling} is related by title, neither the aim nor the
proposed method are connected to ours.

\section{A tunable RMH sampling algorithm} \label{sec:tunable_rmh}
An iteration of the aforementioned RMH algorithm is made up of
$(i)$ proposing a new position $\theta_\text{prop} \sim p_\text{prop}(\cdot \mid \theta_\text{curr})$ from a proposal distribution, typically a Gaussian centered around the current position $\mathcal{N}(\theta_\text{curr}, \Sigma_\text{prop})$; and
$(ii)$ accepting this proposal with probability
\begin{equation}
    p_\text{acc}(\theta_\text{prop} \mid \theta_\text{curr}) =\\ \min\left(1, e^{-\frac{1}{T} (V(\theta_\text{prop};J) - V(\theta_\text{curr};J))}\right)\,,
\end{equation}
where the potential corresponding to the Bayesian canonical ensemble above is
$V(\theta; J) = -\log(\mathcal{L}\pi) - J \cdot \theta$.~\footnote{Please note that the acceptance probability only takes the mentioned
form if the proposal distribution is symmetric,
$p_\text{prop}(\theta_\text{prop} \mid \theta_\text{curr}) =
p_\text{prop}(\theta_\text{curr} \mid \theta_\text{prop})$.  In our case this is
given.}
Considering a potential that depends additionally on the sources $J$ as above, our
aim here is to develop a mechanism that allows tuning of both the temperature
$T$ and the external parameter $J$ while the RMH algorithm runs.

The naive approach would be to implement two schedulers $\sigma_T(t)$ and
$\sigma_J(t)$ that return the current values for the temperature $T$ and
the sources $J$ respectively at simulation step $t$. The acceptance
probability at step $t$ is then modified,
\begin{multline}
    p_\text{accept}(\theta_\text{prop}, \sigma_T(t+1), \sigma_J(t+1) \mid \theta_\text{curr}, \sigma_T(t), \sigma_J(t)) 
    = \min\left(1,  e^{ - \frac{V(\theta_\text{prop}; \sigma_J(t+1))}{\sigma_T(t+1)} + \frac{V(\theta_\text{curr}; \sigma_J(t))}{\sigma_T(t)} }\right)\,,
    \label{eq:accept_prob_dynamic}
\end{multline}
taking into account the change in position $\theta$, temperature $T$ and sources
$J$. The central question for this setup is then how slowly the external
parameters should be changed. For example, if one decreases the temperature very
quickly while $\theta_\mathrm{curr}$ is at a local minimum of $V$, the first
term in the exponential in~\eqref{eq:accept_prob_dynamic} will dominate such
that $p_\text{accept} \to 0$ and once the temperature has reached a sufficiently
low value, most random-walk proposals away from the local minimum are rejected
with high probability, so the sampler effectively becomes stuck. 

An excellent criterion for this is that of \emph{quasi-staticity} from
thermodynamics. It states that the time scale at which a thermodynamic system
should be perturbed externally should be much larger than the time scale at which the
system equilibrates, such that at any point during the induced change
of state, the system remains in thermodynamic equilibrium.~\footnote{In practice
it is enough for the sampler to remain sufficiently close to equilibrium
s.\,t.\@ the macroscopic description remains valid.} Note that this does not
mean that the change of state must happen \enquote{slowly} by the standards of
human perception (e.g., the cyclic processes in heat engines are accurately
described by thermodynamics while they certainly do not look slow to the human
eye). Furthermore, the notion of quasi-staticity is different from
reversibility, which will be discussed below for the use case of sampling
algorithms.

During a normal RMH step, the scale of the proposal step size is typically chosen
such that $V(\theta_\text{prop}) - V(\theta_\text{curr}) \propto T$. Thus, with
the reasonable assumption that $V$ is smooth in $\theta$, $J$ and $T$, one
should choose the scale at which the external parameters are changed such that
\begin{equation}
    V(\theta; \sigma_J(t+1)) - V(\theta; \sigma_J(t)) \ll T\,.
\end{equation}
For a transition from $J = J_i$ to $J_f$, we do so by picking some small
$\epsilon > 0$ and choosing the step size $\Delta J$ such that
$|V(J + \Delta J) - V(J)| \approx |\Delta J \cdot \nabla_J V| = \epsilon T$.
This is achieved by
\begin{equation}
    \Delta J = \min\left(\Delta J_\text{max}, \frac{\epsilon T (J_f - J_i)}{|\nabla_J V \cdot (J_f - J_i)|}\right)\,,
\end{equation}
where we introduced a maximum step size $\Delta J_\text{max}$ to mitigate
division by zero. By choosing $\epsilon$ sufficiently small, the effect on the
acceptance probability remains small by the change in the external parameter $J$
and the sampling continues adapting in real time, remaining close to
thermodynamic equilibrium. Additionally, instead of a single gradient we run
multiple chains at the same time and use the average ensemble gradient $\langle
\nabla_J V \rangle_\mathrm{chains}$ as criterion for the step size.  This is
again inspired by physical heat engines where the combined motion of many
particles drives the macroscopic motion of the piston. For the potential
corresponding to the canonical partition sum \eqref{eq:canZ}, we find $\nabla_J
V = -\theta$.

Similarly, for changes in temperature from $T = T_i$ to $T_f$, we pick $\Delta T
= \epsilon T$ such that we find $\sigma_T(t+1) = (1 + \epsilon) \sigma_T(t)$.
This means that we may write
\begin{equation}
    \sigma_T(t) = (1 + \epsilon)^t T_i\,,
\end{equation}
corresponding to an \emph{exponential scheduler} in simulated annealing
\citep{simulated_annealing}.

While these schedulers are inspired by quasi-staticity, they do not generally
guarantee it. In particular, convergence guarantees for simulated annealing at
low temperatures require a scheduler of the form $\sigma_T(t) \sim c / (1 + \log
t)$ \citep{logarithmic_cooling_schedule}. For our examples the exponential
scheduler works sufficiently well (as we will see by comparing the numerical
results to theoretical predictions) but please be aware that one may have to use
a more robust scheduler for more complicated systems.

\section{Thermodynamic Concepts for MCMC Cycles} \label{sec:thermodynamics}

\begin{figure}[htbp]
    \centering
        \begin{minipage}{0.45\linewidth}
        \centering
        \begin{tikzpicture}[scale=0.8, transform shape, >=Stealth, decoration={markings, mark=at position 0.5 with {\arrow{>}}}]
                        \node (TL) at (0, 4) {$(V_s, T_h)$};
            \node (TR) at (5, 4) {$(V_l, T_h)$};
            \node (BR) at (5, 0) {$(V_l, T_c)$};
            \node (BL) at (0, 0) {$(V_s, T_c)$};

                        \draw[postaction={decorate}] (TL) -- (TR) node[midway, above] {Isothermal expansion};
            \draw[postaction={decorate}] (TR) -- (BR) node[midway, right, align=left] {Isochoric\\cooling};
            \draw[postaction={decorate}] (BR) -- (BL) node[midway, below] {Isothermal compression};
            \draw[postaction={decorate}] (BL) -- (TL) node[midway, left, align=right] {Isochoric\\heating};

            \node at (2.5, -1.5) {\textbf{Stirling cycle}};
        \end{tikzpicture}
    \end{minipage}
    \hfill
        \begin{minipage}{0.45\linewidth}
        \centering
        \begin{tikzpicture}[scale=0.8, transform shape, >=Stealth, decoration={markings, mark=at position 0.5 with {\arrow{>}}}]
                        \node (TL) at (0, 4) {$(J_{(1)}, T_h)$};
            \node (TR) at (5, 4) {$(J_{(2)}, T_h)$};
            \node (BR) at (5, 0) {$(J_{(2)}, T_c)$};
            \node (BL) at (0, 0) {$(J_{(1)}, T_c)$};

                        \draw[postaction={decorate}] (TL) -- (TR) node[midway, above] {Isothermal};
            \draw[postaction={decorate}] (TR) -- (BR) node[midway, right] {Iso-$J$};
            \draw[postaction={decorate}] (BR) -- (BL) node[midway, below] {Isothermal};
            \draw[postaction={decorate}] (BL) -- (TL) node[midway, left] {Iso-$J$};

            \node at (2.5, -1.5) {\textbf{MCMC cycle}};
        \end{tikzpicture}
    \end{minipage}
    \caption{Comparison of a thermodynamic Stirling cycle (left) and its information theoretical analogue, the MCMC cycle (right).}
    \label{fig:cycles_diagram}
\end{figure}
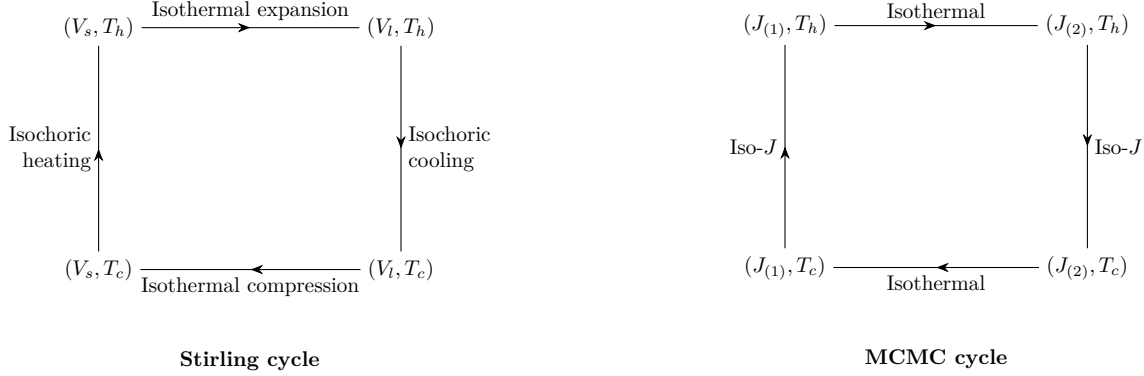

With the canonical partition function defined above \eqref{eq:canZ}, one may
define the free energy $F(T, J) = -T \log Z(T, J)$. Its total differential
$\dd F(T, J) = -S \dd T + \Lambda \cdot \dd J$ then allows the deduction of the
thermodynamic entropy,
\begin{equation} \label{eq:entropy}
S = -\pdv{F}{T}\,,
\end{equation}
as well as the quantity
$\Lambda_i = \partial F / \partial J_i$, which corresponds to volume or
pressure for systems with gases as their working substance in their physical
interpretation, but corresponds as well to the first moment of the posterior
distribution. For our particular case of a canonical partition function, the
entropy $S$ at unit temperature is in fact Shannon's entropy of the posterior
distribution, providing yet another link between thermodynamics and statistics.

The first law of thermodynamics relates the change of the internal energy $\dd U$
to the heat exchanged with the system $\delta Q$ and work performed on it
$\delta W$ for any quasi-static process; it states $\dd U = \delta Q + \delta
W$. Such a process is additionally called reversible if it may be undone without
lasting change to the environment. On a microscopic level, this means that there
is no internal production of entropy and one may write $\delta Q = T \dd S$.
A classic example of an irreversible process would involve the dissipation of
heat through friction which then corresponds to such a production of internal
entropy. Since there is no mechanism in the proposed MCMC cycles that would
correspond to a \enquote{friction term} or similar, we conclude that they should
be generally reversible. 

Thermodynamic cyclic processes in general are a series of changes of state that
form a closed loop in the sense that the final state of the system at
each cycle is the same as the initial one. As a classical example, consider
the Stirling cycle \citep{engineering_thermodynamics}. The idealized Stirling
engine contains an ideal gas and the cycle is made up of four phases that are
described in Figure~\ref{fig:cycles_diagram} (left). In the first phase ($1 \to 2$),
the gas expands isothermally from volume $V_s$ to $V_l$ while in contact with a
hot thermostat at temperature $T_h$. Then ($2 \to 3$) the gas is cooled by
contact with a colder thermostat from $T_h$ to $T_c$ while retaining the same
volume $V_l$. The step ($1 \to 2$) is then mirrored in ($3 \to 4$) by a
compression from $V_l$ to $V_s$ at the low temperature $T_c$. Finally, for ($4
\to 1$), the contact switches again to the hot thermostat to heat up the gas
from $T_c$ to $T_h$ while maintaining the volume $V_s$.

The Stirling cycle and engine have many interesting thermodynamic properties
\citep{engineering_thermodynamics}. For us, it is mainly an inspiration for the
MCMC cycles depicted in Figure~\ref{fig:cycles_diagram} (right). The sources $J$
assume the role of the volume $V$ here. Please note that the question of whether
the sources $J$ are intensive (like the pressure $P$) or extensive (like the
volume $V$) in nature is not straightforward to answer. This has been discussed by
\citet{partition_info}, also by taking the perspective of information geometry.
Since it is not relevant to our discussion, we merely use the Stirling cycle as
an analogy.  (The corresponding \enquote{intensive} thermodynamic cycle where
the changing parameters are $(P, T)$ is in fact the Ericsson cycle.) 

To implement the cycle, we use adaptive ensemble schedulers $\sigma_T$
for the temperature and $\sigma_J$ for the sources. We predefine the target
temperatures $(T_c, T_h)$ and sources $(J_{(1)}, J_{(2)})$ and let the scheduler run
until the desired values are reached for each phase. In contrast to the
volume or pressure in classical thermodynamics, the sources $J$ form a
(co-)vectorial quantity. This means that the behaviour of the sampler depends
non-trivially on the choice of path connecting $J_{(1)}$ and $J_{(2)}$. With our
scheduler, it will always be a straight line s.\,t.\@ the results will only
depend on the choice of end points $J_{(1)}$ and $J_{(2)}$.  The outputs of a run of
such a cycle are the samples $\{\theta_{t,c} \in \mathds{R}^n \mid t = 1, \dots,
N, c = 1, \dots, N_\mathrm{chains}\}$ as well as the trajectories of the
temperatures $\{T_t \in \mathds{R} \mid t = 1, \dots, N\}$ and the sources
$\{J_t \in \mathds{R}^n \mid t = 1, \dots, N\}$.

To analyse the data, we compute the change of the internal energy by finite
differencing, $\dot{U}_t = \langle V(\theta_t; J_t) - V(\theta_{t-1}; J_{t-1})
\rangle_\mathrm{chains}$, the rate of the work performed as $\dot{W}_t = \langle
\Lambda_t \cdot (J_t - J_{t-1}) \rangle_\mathrm{chains}$ and the rate of the
exchanged heat as $\dot{Q}_t = \dot{U}_t - \dot{W}_t$ with the first law of
thermodynamics. All of these may be integrated numerically to yield the change
in internal energy $\Delta U_t$, exchanged heat $\Delta Q_t$, and performed work
$\Delta W_t$ as a function of time.

Regarding heat engines, the most important properties of any cyclic
process are the net work \citep{engineering_thermodynamics}
\begin{equation} \label{eq:net_work_def}
    \Delta W_\text{net} = \oint_\text{one cycle} \hspace{-2em} \dd W
\end{equation}
per cycle performed by the engine as well as the heat input $\Delta
Q_\text{in}$. They allow us to define the efficiency $\eta = \Delta W_\text{net} /
\Delta Q_\text{in}$ which measures how well the engine \enquote{translates} the
input energy in the form of heat into work. We will compute these
quantities for MCMC cycles with different statistical models.

To compare them to a theoretical prediction, we compute the entropy $S$ from the
free energy and its partial derivatives given by the heat capacity $C = T
\partial S / \partial T$ as well as $h^i = T \partial S / \partial J_i$ to find 
\begin{equation}
\dot{Q} = T \dot{S} = C \dot{T} + h \cdot \dot{J}\,. \label{eq:heat_rate}
\end{equation}
If this prediction turns out to correspond to the realised behaviour in our
numerical experiments, it confirms that the construction of the MCMC cycles and
the construction of the thermodynamic theory are correct, in particular the fact
that the cycles are reversible. The latter can also be tested by computing
$\Delta S_\mathrm{diss} = \oint \delta Q / T$ which should be compatible with
zero within tolerance.

\section{Application} \label{sec:application}

\subsection{Gaussian Toy Model} \label{sec:gauss}

\begin{figure}[htbp]
    \centering
    \includegraphics[width=.48\textwidth]{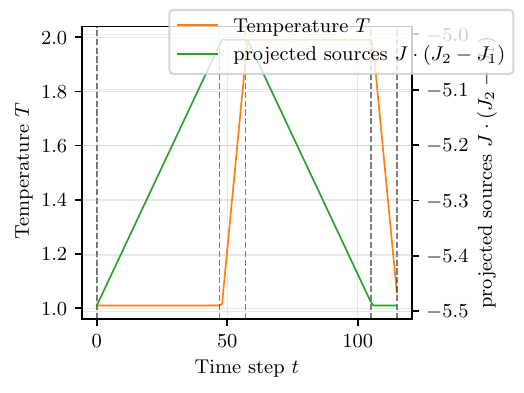}
    \includegraphics[width=.48\textwidth]{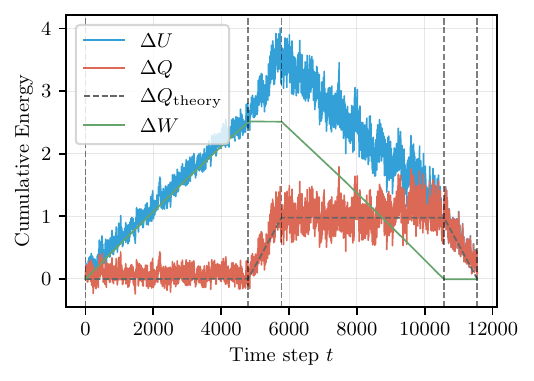}
    \caption{Left: Trajectories of the parameters $T$ and $J$ during an MCMC cycle in an
        anisotropic Gaussian model as a function of simulation time.
        Right: Energy changes during an MCMC cycle in an anisotropic Gaussian
        model along with a theoretical prediction for the heat change. The dotted
        vertical lines differentiate the four phases of the cycle.}
    \label{fig:anisotropic_gaussian_parameter_trajectories}
    \label{fig:anisotropic_gaussian_energy}
\end{figure}

As a first example, consider a Gaussian likelihood
\begin{equation} \label{eq:gaussian_likelihood}
\mathcal{L} \propto \exp(-\frac{1}{2} F_{ij} \theta^i \theta^j)\,,
\end{equation}
with Fisher information $F$ and flat prior $\pi \equiv 1$. (Any constant
normalisation factors would drop out in the following calculations, so we ignore
them.) For such a model one may compute the partition function analytically 
\begin{equation}
    Z(T, J) = \sqrt{\frac{(2\pi T)^n}{\det F}} \exp\left(\frac{1}{2T} F^{ij} J_i J_j\right)\,,
\end{equation}
with the components $F^{ij}$ of the inverse Fisher matrix such that one finds
the entropy \eqref{eq:entropy} 
\begin{equation}
    S(T, J) = \frac{n}{2} \log(2\pi T) + \frac{n}{2} - \frac{1}{2} \log \det F\,,
\end{equation}
and thus the caloric coefficients $C(T, J) = n/2$ and $h^i(T, J) =
0$ \citep{partitionfunction101}. The theoretical rate at which the heat is
exchanged with the system is 
\begin{equation} \label{eq:q_dot_theory_gauss}
\dot{Q} = \frac{n}{2} \dot{T}\,.
\end{equation}
This result
means that during the iso-$J$ phase of the cycle the amount $\Delta Q =
\frac{n}{2} (T_h - T_c)$ is exchanged with the environment. For both
phases, it is the same amount with a flipped sign, implying $\Delta Q_\text{net}
= 0$. And since $\Delta W_\text{net} = \Delta U_\text{net} - \Delta Q_\text{net}
= 0$, there is also no work performed by the system. The internal energy
is the same after the cycle as before, since the system returns to its
initial state. These results are calculated in more detail in
Appendix~\ref{app:toy_model_thermodynamics}.

To test whether this theoretical prediction holds in practice, we run multiple
cycles with this system. In this numerical experiment in two dimensions we set $F = \text{diag}(1,
10)$, $J_{(1)} = (5.5, 0)$, $J_{(2)} = (5, 0)$, $T_c = 1.0$, $T_h = 2.0$.

\begin{figure}[htbp]
    \centering
    \includegraphics[width=.49\linewidth]{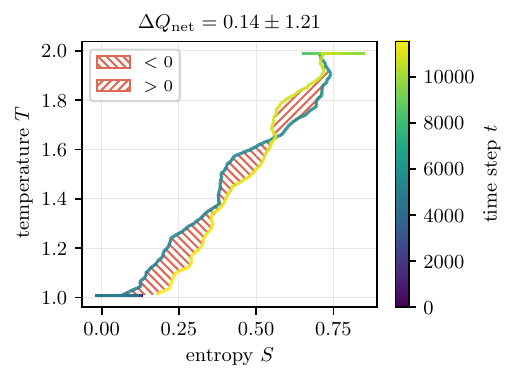}
    \includegraphics[width=.49\linewidth]{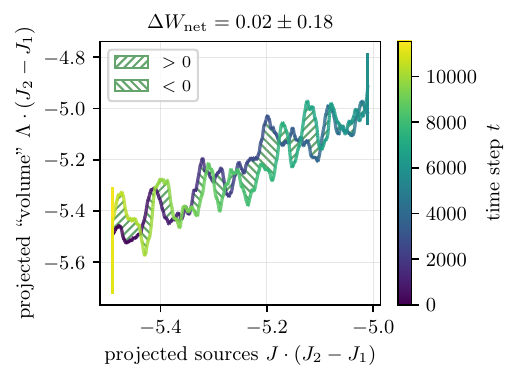}
    \caption{Illustration of the net exchanged heat (left) and net performed work
    (right) during one cycle. Since the heat and work forms are given by
    $\delta Q = T \dd S$ and $\delta W = \Lambda \cdot \dd J$ respectively, the
    areas enclosed by the trajectory correspond to the net exchanged heat and
    net performed work during one cycle.  To reduce noise, the data are
    smoothed with a moving average with periodic boundary conditions of
    window size 200 simulation steps.}
    \label{fig:anisotropic_gaussian_heat_work}
\end{figure}

In Figure~\ref{fig:anisotropic_gaussian_parameter_trajectories} (left), the
trajectories of temperature and sources are plotted. One may observe that the
exponential temperature scheduler does indeed deserve its name. While the
iso-$J$ phases look quite exponential in the temperature $T$, in contrast to
this, the isothermal trajectories look rather linear in $J$. This is because the
average gradient $\langle \nabla_J V \rangle_\mathrm{chains}$ is too small such
that the scheduler almost always takes the maximum step size $\Delta
J_\text{max}$.

In Figure~\ref{fig:anisotropic_gaussian_energy} (right) we plot the observed energy
changes $\Delta U, \Delta W, \Delta Q$ along with the theoretical predictions
$\Delta Q_\text{theory}$ made above. To reduce the considerable amount of noise
we average over $10^2$ chains.  The work $\Delta W$ has much less noise than the
energy change $\Delta U$, since for an individual chain $\dot{W} \propto \theta
\cdot \dot{J}$ (where $\dot{J}$ is the same for all chains) while $\dot{U} \propto
\theta \cdot \dot{\theta}$. Thus we expect the noise on the energy change to be
much larger compared to the noise of the work change.

It is clearly visible that $\Delta Q = 0$ in the isothermal phases, while
$\Delta W = 0$ in the iso-$J$ phases, confirming the theoretical predictions.
Also, the noise of the energy change $\Delta U$ is greater for $T = T_h$
compared to $T = T_c < T_h$, which also makes sense since the sampling particles
have more freedom to oscillate around the expectation value of $\theta$ in
this phase. For comparison, we also plot the theoretical predictions
\eqref{eq:q_dot_theory_gauss} for the heat change $\Delta Q_\text{theory}$,
which agree nicely with the observations. The fact that $\Delta U_\mathrm{net}
= 0$ after a full cycle is to be expected for any model as described before.
Additionally, this plot conveys that the system is not performing any net work,
since $\Delta W_\mathrm{net} = \Delta Q_\mathrm{net} = 0$ after a full cycle.

The exchanged heat and work may also be read off
Figure~\ref{fig:anisotropic_gaussian_heat_work}. On the right, we plot the
temperature $T$ against the entropy $\dd S = \delta Q/T$ (again we apply a
moving average over 200 simulation steps to reduce noise). Since the net heat
exchange over one cycle is given by $\Delta Q_\text{net} = \oint T \dd S$, it is
given in the plot as the area enclosed by the trajectory. On the left of
Figure~\ref{fig:anisotropic_gaussian_heat_work}, we made a similar plot for the
quantity $\Lambda = \nabla_J H$ against the sources $J$. Since the scheduled
source $\sigma_J(t)$ always lies on a line connecting $J_{(1)}$ and $J_{(2)}$, we
equivalently consider the projection of $J$ and $\Lambda$ on their difference
$J_{(2)} - J_{(1)}$. Again one observes that the net work $\Delta W_\text{net} = \oint
\Lambda \cdot \dd J$ vanishes.  In Figure~\ref{fig:anisotropic_gaussian_summary} (left)
all energy measurements concerning this experiment are summarised. Again, the 
higher temperature in phase $(3 \to 4)$ leads to larger fluctuations and the net
performed work and exchanged heat are zero within tolerance. To demonstrate
that one may arbitrarily reduce the noise of the heat, we plot $\Delta Q_\mathrm{net}$
against the number of cycles in Figure~\ref{fig:gaussian_q_net_vs_n_cycles} (right),
maintaining compatibility with zero up to $\sim 10^2$ cycles.

\begin{figure}[tbp]
    \centering
    \includegraphics[width=.48\textwidth]{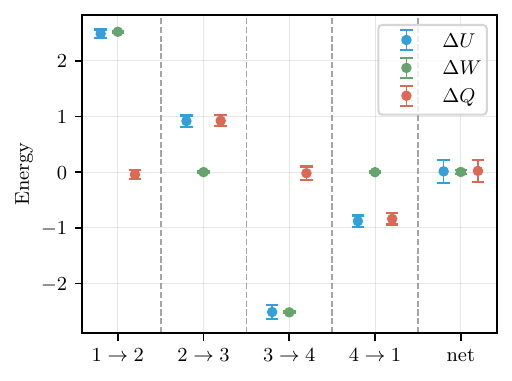}
    \includegraphics[width=.48\textwidth]{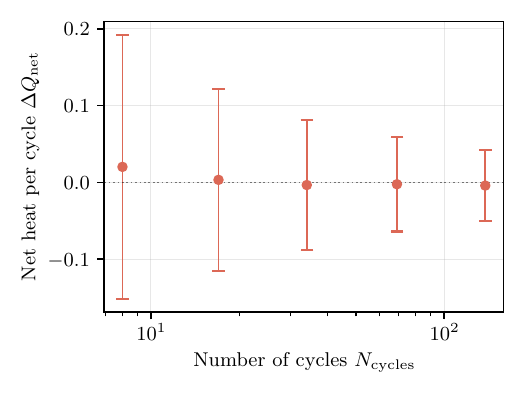}
    \caption{Left: Total energy changes per phase of an MCMC cycle in an anisotropic
        Gaussian model. The results are averaged over $10^2$ chains and six
        analysed cycles (after excluding two warm-up cycles).
        Right: Net heat per cycle $\Delta Q_{\mathrm{net}}$ as a function of the
        number of cycles for the anisotropic Gaussian model.}
    \label{fig:anisotropic_gaussian_summary}
    \label{fig:gaussian_q_net_vs_n_cycles}
\end{figure}

\begin{figure}[htbp]
    \centering
    \includegraphics[width=.48\textwidth]{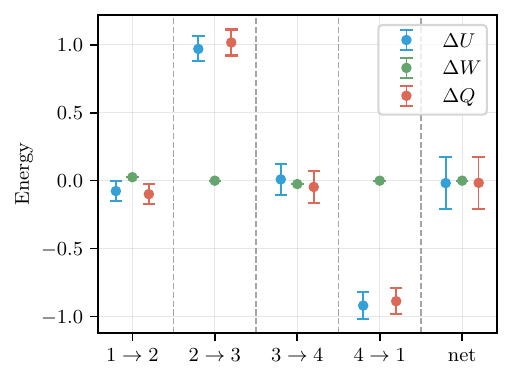}
    \caption{Results of an MCMC cycle in an anisotropic Gaussian model with
        different target sources. The performed work in phases $(1 \to 2)$ and $(3 \to 4)$ is much less than before.}
    \label{fig:anisotropic_gaussian_J_along}
\end{figure}

\paragraph{Gaussian Toy Model with different cycle trajectory}
As a second example, we consider the same model as above
\eqref{eq:gaussian_likelihood}, but this time we choose the sources $J_{(1)} = (0,
5.5)$, $J_{(2)} = (0, 5)$ to lie on a different axis. As above, the Gaussian is
anisotropic with the Fisher matrix $F = \text{diag}(1, 10)$ such that the
typical set has the shape of an ellipsoid which is elongated in
$\theta_1$-direction. In the previous example, we chose the sources along the
elongated direction ($J$ in $\theta_1$). This time, they are pointed across
it ($J$ in $\theta_2$).

The consequence of this may be observed in
Figure~\ref{fig:anisotropic_gaussian_J_along}, which corresponds to
Figure~\ref{fig:anisotropic_gaussian_summary} (left) for the above experiment. The work
performed in the isothermal phases $(1 \to 2)$ and $(3 \to 4)$ of the cycle is
much smaller. This was to be expected since the system has to work against much
less resistance in this case as the potential varies less in this direction.

\subsection{Rosenbrock Toy Model} \label{sec:rosenbrock}
As an example of a non-linear and consequently non-Gaussian model, we consider
the Rosenbrock likelihood
\begin{equation} \label{eq:rosenbrock_def}
\mathcal{L} \propto \exp(-\frac{1}{2} \left( (a \theta^0)^2 + (b (\theta^1 - (c \, \theta^0)^2))^2 \right))\,,
\end{equation}
with $a = 2$, $b = 10$ and $c = 0.5$ and a flat prior $\pi \equiv 1$. We use
$T_c = 1.0$ and $T_h = 3.0$ for both cycle trajectories. For the zero-work
trajectory we set $J_{(1)} = (6, 0)$ and $J_{(2)} = (5, 0)$, whereas for the
nonzero-work trajectory we set $J_{(1)} = (0, -4)$ and $J_{(2)} = (0, -6)$.
Again, the partition sum may be computed analytically. Substituting
$u=\theta^1-c^2(\theta^0)^2$, we are left with two Gaussian integrals that yield
\begin{equation}
Z(T,J)
=
\frac{2\pi T}{b\sqrt{a^2-2c^2J_1}}
\exp\left[
\frac{1}{2T}
\left(
\frac{J_1^2}{b^2}
+
\frac{J_0^2}{a^2-2c^2J_1}
\right)
\right],
\end{equation}
for $a^2-2c^2J_1>0$ (see Appendix~\ref{app:toy_model_thermodynamics} for a 
more detailed calculation). The corresponding entropy is
\begin{equation}
S(T,J)
=
1+\log(2\pi T)-\log b
-\frac12\log(a^2-2c^2J_1),
\end{equation}
leading to
\begin{equation} \label{eq:q_dot_theory_rosenbrock}
C=T\frac{\partial S}{\partial T}=1,
\qquad
h^0=T\frac{\partial S}{\partial J_0}=0,
\qquad
h^1=T\frac{\partial S}{\partial J_1}
=
\frac{Tc^2}{a^2-2c^2J_1}
\quad
\rightarrow
\quad
\dot Q
=
\dot T
+
\frac{Tc^2}{a^2-2c^2J_1}\dot J_1 .
\end{equation}
Since the entropy has an explicit dependence on both $J$ and $T$
in contrast to the purely Gaussian case, it is to be
expected that the exchanged heat during the iso-$J$ phases depends on $J$,
thus $\Delta Q_{12} \neq -\Delta Q_{34}$. This means that $\Delta Q_\text{net}
\neq 0$ and consequently a non-zero amount of work $\Delta W_\text{net} \neq 0$
is performed during a cycle. However, as may be seen in
\eqref{eq:q_dot_theory_rosenbrock} this is only the case if component $1$ of the
sources varies during the cycles. If merely component $0$ varies, there will be
no net work output.

To improve the results of this experiment and test our model on a different method,
this time we use Hamiltonian Monte Carlo \citep{hmc_original}. Essentially, this
means that the proposal in step $(i)$ of our tunable RMH algorithm in
section~\ref{sec:tunable_rmh} is modified to use the endpoint of an integrated HMC
trajectory. By this we hope to increase the speed with which the sampler adapts
to the changes in the external parameters; otherwise, the algorithm remains
unaltered.  After a series of cycles with this model, the data are analysed as
above for the Gaussian case and the results are summarised in
Figure~\ref{fig:rosenbrock_summary}. While the plots display similarities to the
former case, the lower row conveys that the net performed work as well as the
net heat exchanged are non-zero this time.  While the overall change in internal
energy is zero after every cycle, the heat input rises gradually while
accumulated performed work declines equivalently after every cycle (so in this
case, we have an information theoretic heat pump rather than a heat engine).

\begin{figure}[p]
    \centering
    \includegraphics[width=.48\linewidth]{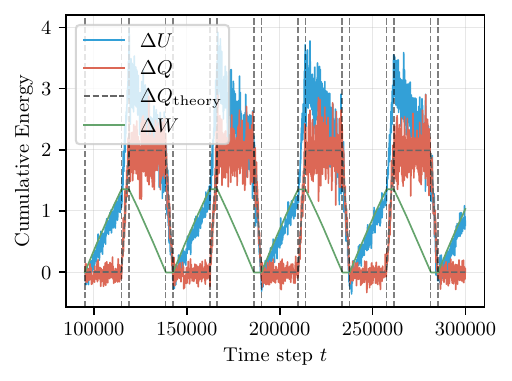}
    \includegraphics[width=.48\linewidth]{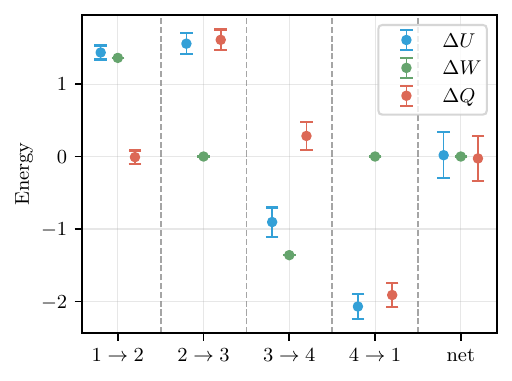}
    \hspace{.7em}
    \includegraphics[width=.48\linewidth]{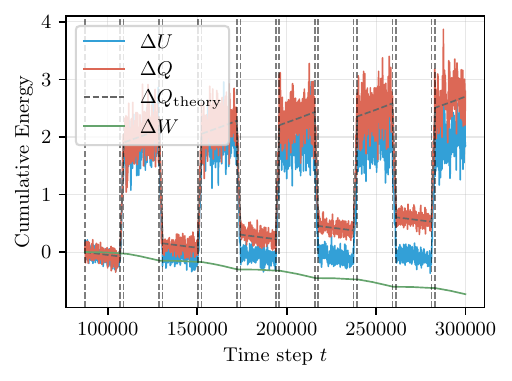}
    \includegraphics[width=.48\linewidth]{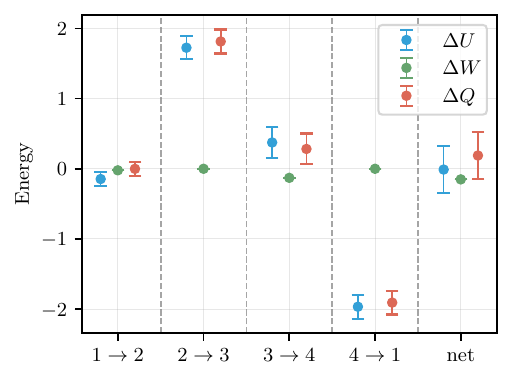}
    \captionsetup{type=figure,hypcap=false}
    \caption{Results for the Rosenbrock model corresponding to Figures
    \ref{fig:anisotropic_gaussian_energy} and
    \ref{fig:anisotropic_gaussian_summary} for the anisotropic Gaussian model.
    In contrast to before, a non-zero amount of net work is performed by the
    system if the sources are changed along direction $1$ (lower row).}
    \label{fig:rosenbrock_summary}
\end{figure}

\begin{figure}[p]
    \centering
    \includegraphics[width=.48\textwidth]{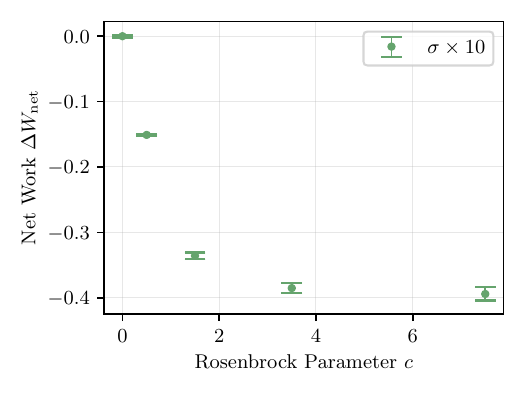}
    \caption{Net performed work as a function of the parameter $c$ in the Rosenbrock model. The more non-Gaussian the model becomes, the greater is the work output.}
    \label{fig:rosenbrock_net_work_vs_c}
\end{figure}

To further study this phenomenon, we ran this experiment for different
choices of the parameter $c$ in the Rosenbrock likelihood. In
Figure~\ref{fig:rosenbrock_net_work_vs_c}, the net performed work is plotted
against it. This sweep uses the nonzero-work cycle trajectory and varies $c$
from 0 to 8. At $c=0$ the model is Gaussian, thus the output is zero within
tolerance. The greater the parameter becomes, the more work is produced.

\subsection{General non-Gaussian model}
Having demonstrated that MCMC cycles with purely Gaussian models always produce
vanishing net work, $\Delta W_{\mathrm{net}} = 0$, we would like to show the
converse direction.

\begin{proposition}
Consider any non-Gaussian model $\mathcal{L}(\theta) \propto e^{-V(\theta)}$
with a smooth potential $V$ that has a non-degenerate global minimum. Construct
an MCMC cycle in $(T,J)$ space as described in Sec.~\ref{sec:thermodynamics}
where the path $J_{12}(t)$ connecting $J_{(1)}$ and $J_{(2)}$ in the isothermal
phase $(1 \to 2)$ is the inverse of the path used in phase $(3 \to 4)$. It is
possible to choose this cycle such that the sampler will perform a non-zero
amount of net work as defined in \eqref{eq:net_work_def} along it.
\end{proposition}

\begin{proof}
Without loss of generality, place the potential's global minimum at $V(\theta =
0)=0$, and formally expand around it,
\begin{equation}
    V(\theta)
    =
    \frac{1}{2} F_{ij} \theta^i \theta^j
    +
    \widetilde{V}(\theta),
\end{equation}
where $F_{ij}$ is the Hessian evaluated at the minimum and
$\widetilde{V}(\theta)$ contains all higher-order terms. With this, we
may write the partition function as
\begin{align}
    Z(T,J)
    &=
    \int \dd^n \theta\,
    \exp\left[
        -\frac{1}{T} V(\theta)
        +
        \frac{1}{T} J_i \theta^i
    \right]
    =
    \int \dd^n \theta\,
    \exp\left[
        -\frac{1}{2T} F_{ij}\theta^i\theta^j
    \right]
    \exp\left[
        \frac{1}{T}J_i\theta^i
        -
        \frac{1}{T}\widetilde{V}(\theta)
    \right]
    \\
    &=
    \sqrt{\frac{(2\pi T)^n}{\det F}}\,
    \left\langle
        \exp\left[
            \frac{1}{T}J_i\theta^i
            -
            \frac{1}{T}\widetilde{V}(\theta)
        \right]
    \right\rangle_{\theta \sim \mathcal{N}(0,\,T F^{-1})}.
\end{align}
Now, we may use a well-known identity for Gaussian measures (see
Appendix~\ref{app:gaussian_measure_identity} for a short motivation) to write
\begin{align} 
    &Z(T,J) \nonumber \\
    &=
    \left.
    \sqrt{\frac{(2\pi T)^n}{\det F}}\,
    \exp\left(
        \frac{T}{2} F^{ij}\partial_i\partial_j
    \right)
    \exp\left[
        \frac{1}{T}J_k\theta^k
        -
        \frac{1}{T}\widetilde{V}(\theta)
    \right]
    \right|_{\theta = 0} \label{eq:qft_expr} \\
    &=
    \sqrt{\frac{(2\pi T)^n}{\det F}}
    \left[
        1
        +
        \frac{1}{2T}F^{ij}J_iJ_j
        -
        \frac{1}{2T}F^{ij}
        \bigl(\partial_i\widetilde{V}\bigr)
        \bigl(\partial_j\widetilde{V}\bigr)
        -
        \frac{1}{2 T}F^{ij}
        \partial_i\partial_j\widetilde{V}
        - \frac{1}{2 T} F^{ij} J_i (\partial_j \widetilde{V})
        + \mathcal{O}(T^{-2})
    \right]_{\theta=0}
    \\
    &=
    \underbrace{
    \sqrt{\frac{(2\pi T)^n}{\det F}}\,
    \exp\left[
        \frac{1}{2T}F^{ij}J_iJ_j
    \right]}_{=:Z_\mathrm{G}}
    \underbrace{
    \left[
        1
        -
        \frac{1}{2T}F^{ij}
        \bigl(\partial_i\widetilde{V}\bigr)
        \bigl(\partial_j\widetilde{V}\bigr)
        -
        \frac{1}{2 T}F^{ij}
        \partial_i\partial_j\widetilde{V}
        - \frac{1}{2 T} F^{ij} J_i (\partial_j \widetilde{V})
        +
        \mathcal{O}(T^{-2})
    \right]_{\theta=0}}_{=:Z_\mathrm{NG}} \, .
\end{align}
With this, consider the work form
$\dd W = \Lambda \cdot \dd J$, with the quantity
\begin{equation}
\Lambda^i
= \partial^i F
= -T \partial^i \log Z 
= -T \partial^i \log Z_{\mathrm{G}} -T \partial^i \log Z_{\mathrm{NG}}
=: \Lambda^i_{\mathrm{G}} + \Lambda^i_{\mathrm{NG}}\, ,
\end{equation}
where we denote $\partial^i := \partial / \partial J_i$.
For the Gaussian part, we find
$\Lambda^i_{\mathrm{G}}(J) = -F^{ij}J_j$, to be \emph{independent} of the
temperature. This means that
\begin{equation}
    \Delta W_{12}
    =
    \int_{J_{12}(t)} \dd J \cdot \Lambda_{\mathrm{G}}(J)
    =
    -
    \Delta W_{34} \quad \rightarrow \quad \Delta W^\mathrm{G}_\mathrm{net} = 0\, ,
\end{equation}
for any choice of path $J_{12}(t)$. $\Delta W_{23} = \Delta W_{41} = 0$ since these
edges of the cycle are iso-$J$.  At the same time, as soon as a single term in
the series expansion of $Z_\mathrm{NG}$ becomes nonzero, the partition sum will
be non-Gaussian, $Z_\mathrm{NG} \neq 1$. Thus, $\Lambda_\mathrm{NG}$ will depend
on the temperature s.\,t.\@ the two functions $\Lambda_{\mathrm{NG}}(T_c,\cdot)$
and $\Lambda_{\mathrm{NG}}(T_h,\cdot)$ are distinct and one may thus find a path
$J_{12}(t)$ that fulfils
\begin{equation}
    \Delta W_{12}
    = \int_{J_{12}(t)} \dd J \cdot \Lambda(T_h,J)
    \neq \int_{J_{12}(t)} \dd J \cdot \Lambda(T_c,J)
    = -\Delta W_{34}
    \quad \rightarrow \quad
    \Delta W^\mathrm{NG}_\mathrm{net} \neq 0\, .
\end{equation}
\end{proof}
To reiterate in more detail, an MCMC cycle with a Gaussian model will always
produce $\Delta W_\mathrm{net} = 0$ independent of the choice of cycle
trajectory in $(T,J)$ space. For a non-Gaussian model, it is always possible to
find a cycle trajectory s.\,t.\@ the resulting cycle produces $\Delta
W_\mathrm{net} \neq 0$. This is precisely what we have seen in our examples
above. For the Gaussian in Section~\ref{sec:gauss} we found $\Delta
W_\mathrm{net} = 0$ for both cycle trajectories (that span the space of all
possible trajectories in that case). The non-Gaussian Rosenbrock likelihood in
Section~\ref{sec:rosenbrock}, on the other hand, exhibited $\Delta
W_\mathrm{net} = 0$ along the first chosen trajectory while we found $\Delta
W_\mathrm{net} \neq 0$ for the second trajectory.

The effect of anharmonicities in an otherwise quadratic potential is well
known in solid state physics, where it explains the thermal expansion of
solids~\citep{solid_state_theory}. A rudimentary realisation of a heat engine
might be a solid that is subjected to stresses and changes in temperature,
emulating a cycle in $J$ and $T$. If the atoms were bound in a mutual potential
of quadratic shape, they would show vibrations about the equilibrium position.
In an anharmonic potential, however, the centre of vibration would be shifted
away from the minimum of the potential.  This effectively increases the average
distance between the atoms, leading to thermal expansion. From a thermodynamic
point of view, the minimum in free energy determines the average distance
between the atoms. Then, thermal expansion can perform work against external
stresses, in contrast to the case of a quadratic potential, where thermal
expansion does not occur and mechanical work is not performed. This completes
the argument why a quadratic potential representing Gaussian distributions is
ineffective as a heat engine and leads to a zero net work output.

Effectively, the mechanical work performed in a Stirling-type MCMC cycle
reflects the deviation from Gaussianity of the posterior distribution, and the
net work output of the cyclic MCMC engine serves as a measure of non-Gaussianity that
is independent from the assumption of any particular functional shape. Please
note, however, that this is not computationally competitive with established
methods for such measurements. For instance, one might simply sample $10^3$
samples from the above Rosenbrock model \eqref{eq:rosenbrock_def} with RMH or
HMC and then compare the distribution of the samples to the empirical Gaussian
$\mathcal{N}(\hat{\mu}, \hat{\Sigma})$.  This would be much cheaper than our
method since we need $\sim 10^5$ steps of $10^2$ chains to reduce the
uncertainty of our measurements to a reasonable amount.
\subsection{Non-Gaussian cosmological inference problem}

\begin{figure}[htbp]
    \centering
    \includegraphics[width=.48\textwidth]{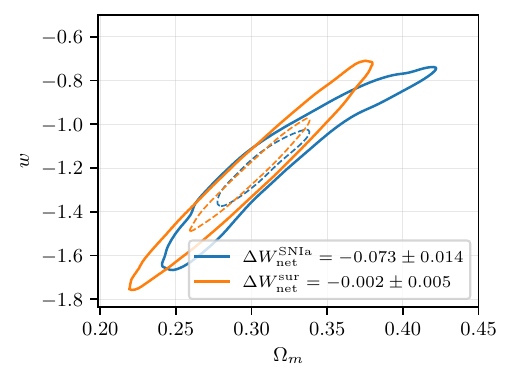}
    \caption{Kernel density estimation of samples from the SN Ia posterior (blue)
    and its Gaussian approximation (orange). The apparent non-Gaussianity is reflected
    in the fact that the MCMC cycles performed non-zero net work with the former
    but not with the latter.}
    \label{fig:sn1a}
\end{figure}

As a topical Bayesian inference problem we consider the determination of
cosmological parameters in supernova cosmology \citep{pantheonshoes}, where the
relationship between redshift $z$ and apparent magnitudes $m$ is probed.  The
Hubble function for a dark energy-dominated Universe is given by
\begin{equation}
    H^2(z) = H_0^2 \left( \Omega_m (1+z)^3 + (1 - \Omega_m) (1+z)^{3(1+w)} \right)\,.
\end{equation}
from which the luminosity distance is obtained numerically by solving
\begin{equation}
    \frac{\dd d_L}{\dd z} = \frac{d_L}{1+z} + \frac{c(1+z)}{H(z)}\,,
\end{equation}
from which one may then compute the prediction for the observed apparent magnitude,
\begin{equation}
    m_\mathrm{pred} = M + 5 \log_{10} d_L + 10\,,
\end{equation}
with the absolute magnitude $M$ as a nuisance parameter. The likelihood is
Gaussian in the data,
\begin{equation} \label{eq:sn1a_likelihood}
    \mathcal{L} \propto \exp\left(-\frac{1}{2} (m^i_{\text{obs}} - m^i_{\text{pred}}(\theta)) C_{ij} (m^j_{\text{obs}} - m^j_{\text{pred}}(\theta))\right)\,,
\end{equation}
with the inverse data covariance $C_{ij}$. To make
this application comparable to the two examples above, we profile the
likelihood, fixing the absolute magnitude $M = -19.25$ and the Hubble constant
$H_0 = 73.28\,\mathrm{km}/\mathrm{s}/\mathrm{Mpc}$, thus obtaining a two-dimensional inference problem with $\theta =
(\Omega_m, w)$.

Since the above likelihood is quite expensive to evaluate, we ran the experiment
with a Gaussian surrogate first to tune the hyperparameters of the cycle. For
this, we approximated the likelihood around the maximum likelihood estimate as
\begin{equation} \label{eq:sn1a_surrogate}
    \mathcal{L} \propto \exp\left(-\frac{1}{2} (\theta^i - \theta^i_\text{MLE}) F_{ij}(\theta_\text{MLE}) (\theta^j - \theta^j_\text{MLE})\right)\,,
\end{equation}
with the Fisher information $F$.~\footnote{This is done only to save time, one
could equally have used the full likelihood directly.} For the MCMC cycle, we
chose $J_{(1)} = (0.1, 1.0)$, $J_{(2)} = (0.15, 1.5)$, $T_c = 1.0$, $T_h =
10.0$. We ran the cycles both with the surrogate and the original model and
report the resulting net works in Figure~\ref{fig:sn1a} along with kernel
density estimates of samples from the posteriors. The value of 
$\Delta W_\mathrm{net} = -0.073 \pm 0.014$ clearly reflects the non-Gaussianity of the model.
Since we could not predict the heat exchange this time, we compute
$\Delta S_\mathrm{diss} = \oint \delta Q / T$ to make sure that the cycle
remained reversible. We find $\Delta S_\mathrm{diss} = 0.1 \pm 0.4$ to be
compatible with zero, confirming reversibility and the validity of the result.

\section{Conclusion} \label{sec:conclusion}
In a proof of concept, this paper demonstrates that cyclic processes for MCMC
algorithms may be constructed theoretically, realised in practice, and allow the
quantification of non-Gaussianities of the posterior distributions in Bayesian
inference. From a conceptual point of view, the MCMC cycles are based on the
realisation that Bayesian inference problems and MCMC algorithms may be
described as thermodynamic systems with the methods from statistical physics,
using the notions of partition functions, their associated thermodynamic
potentials, their state variables and finally net heat exchanged and mechanical
work performed.  Our main technical contributions are the adaptive ensemble
schedulers that facilitate a tunable version of the
Rosenbluth-Metropolis-Hastings and Hamiltonian Monte Carlo algorithms. We apply
the notion of quasi-staticity from thermodynamics to ensure that the sampling
algorithms are affected at timescales large enough to ensure they remain close
enough to thermal equilibrium throughout any change of state. With these
schedulers we implement MCMC cycles for a Bayesian canonical ensemble made up of
four phases alternately varying the temperature $T$ and the sources $J$ inspired
by the Stirling cycle from thermodynamics.  Such cycles can be defined for a
broad class of statistical models for which the relevant moments are finite and
as long as the chosen schedulers allow the sampler to remain sufficiently close
to thermodynamic equilibrium throughout the different phases.  They generalise
numerical methods such as simulated annealing by $(i)$ introducing cyclic
changes of state and by $(ii)$ altering additional parameters, in our case the
sources $J$ used in statistical physics for the determination of cumulants or
moments of the posterior distribution.

We first test them on an anisotropic Gaussian, confirming all predictions made
by a corresponding thermodynamic theory. Our most interesting fundamental
insight surely is that the MCMC cycles can produce non-zero net work output if
and only if the model under consideration is non-Gaussian. This was first found
by comparing the Gaussian example to that of a Rosenbrock likelihood; we then
verified our claim analytically for arbitrary non-Gaussian posteriors. To be
more precise, in case of a non-Gaussian posterior, it is always possible to find
a cycle trajectory along which the MCMC cycle will perform non-zero net work
while the net-work will be zero for any trajectory for a Gaussian posterior.
Thus, the work performed in a cycle by an MCMC cycle may serve as a
(computationally very expensive) measure of non-Gaussianity of the inferred
posterior distribution, independent of its particular functional form.  We
verify the applicability of our method to a practical example from supernova
cosmology, diagnosing a non-Gaussianity in the model. To put our findings into a
physical context, we give an example from solid state physics with an analogous
phenomenology.

The most obvious idea for further study is to apply these MCMC cycles to a wider
range of statistical models. In particular, it would be interesting to consider
higher-dimensional inference problems where the question of the cycle trajectory
would become more important. Here, one could even study cycle trajectories only
in the sources $J$ at constant temperature which should generally not perform 
any net work.  A more technical study would be to use different sampling
algorithms. As we have demonstrated, it is straightforward to use the adaptive
schedulers with any canonical methods (in the thermodynamic sense) such as
adaptive RMH \citep{adaptive_rmh}, NUTS \citep{nuts}, RMHMC \citep{rmhmc}, MALA
\citep{mala} etc. Algorithms that are based on different thermodynamic ensembles
such as the microcanonical MAMS \citep{mams} or the macrocanonical Avalanche
algorithm \citep{partition_macro} would require more care since the
thermodynamic theory would need to be adapted.  Probably, it is best to choose
the algorithm that would be best suited to sample from a given likelihood even
without running any cycles. Our central result will hold for any canonical
sampling algorithm.

In conclusion, the application of concepts from statistical physics to canonical
MCMC methods has again produced both fruitful technical and fundamental
insights, providing a new angle on their thermodynamic nature.

\section*{Acknowledgements}

\paragraph{Funding information}
This work was supported by the Deutsche Forschungsgemeinschaft (DFG, German
Research Foundation) under Germany's Excellence Strategy EXC 2181/1 - 390900948
(the Heidelberg STRUCTURES Excellence Cluster). We acknowledge the usage of the
AI-clusters {\em Tom} and {\em Jerry} funded by the Field of Focus 2 of
Heidelberg University. HvC is supported by the Konrad Zuse School of Excellence
in Learning and Intelligent Systems (ELIZA) through the DAAD programme Konrad
Zuse Schools of Excellence in Artificial Intelligence, sponsored by the Federal
Ministry of Education and Research. 

The authors would like to thank Rebecca Maria Kuntz for the help in typesetting,
Benedikt Schosser for help with plots, and both of them for insightful
discussions and helpful comments.

\bibliography{references}
\bibliographystyle{tmlr}

\appendix

\section{Einstein summation convention}
\label{app:einstein_summation}
In this paper, we use the Einstein summation convention. The general
idea is to work with components of vectors, covectors and tensors rather than
with the objects themselves. Instead of giving a fully formal definition, we
summarise the rules that are relevant for the calculations in this paper.

\begin{itemize}
\item Components of vectors are written with contravariant (upper) indices,
for example $\theta = (\theta^1,\ldots,\theta^n)$.
\item Components of covectors are written with covariant (lower) indices,
for example $J = (J_1,\ldots,J_n)$.
\item Components of bilinear forms are written with two covariant indices, for
example \(F_{ij}\). The components of the inverse bilinear form are written with
two contravariant indices, for example \(F^{ij}\).

\item Whenever an index appears exactly twice in a term, once in the upper and once
in the lower position, it is implicitly summed over. For example,
\begin{equation}
    J_i\theta^i
    =
    \sum_{i=0}^{n-1} J_i\theta^i .
\end{equation}
Such repeated indices are called dummy indices and may be renamed freely, e.g.
    $J_i\theta^i
    =
    J_k\theta^k$.
(In terms of computer science, their \enquote{scope} is limited to the term in
which they appear.) This is the actual Einstein summation convention.
\item Indices which appear only once are free indices. An expression with one free
index represents the components of a vector or covector. Expressions in which
the same index appears more than twice are not used in this convention and will
be avoided.

\item Taking derivatives changes the index position. In this paper we use the
notation
\begin{equation}
    \partial^i F
    :=
    \frac{\partial F}{\partial J_i},
    \qquad
    \partial_i \widetilde V
    :=
    \frac{\partial \widetilde V}{\partial \theta^i}.
\end{equation}
Thus differentiation with respect to a covariant component \(J_i\) produces an
upper index, while differentiation with respect to a contravariant component
\(\theta^i\) produces a lower index.

\item Taking the derivative of a vector component w.\,r.\,t.\@ another yields
a Kronecker delta, for example
\[
\partial_i \theta^j = \delta^j_i\,.
\]

\end{itemize}

Let us give a few examples of how standard matrix notation is translated into
index notation.
\begin{itemize}
\item The symmetry of a matrix $F^\top=F$ is written as $F_{ij}=F_{ji}$.
\item The inverse relation \(FF^{-1}=F^{-1}F=\mathds{1}\) becomes
$
    F_{ij}F^{jk}
    =
    F^{ij}F_{jk}
    =
    \delta^i_k
$.

\item Since we work with components, the ordering of factors in a term is
irrelevant.  For example,
\begin{equation}
    J^\top F^{-1}J
    =
    J_iF^{ij}J_j
    =
    F^{ij}J_iJ_j
    =
    J_jJ_iF^{ij}.
\end{equation}
\end{itemize}

As a more complicated example, consider how the first non-trivial term in the
operator
$
    \exp\left(
        T / 2 \, F^{ij}\partial_i\partial_j
    \right)
$
acts on the second-order term in the expansion of \(\exp(J_k\theta^k/T)\) in
\eqref{eq:qft_expr}. We obtain
\begin{align}
    \frac{T}{2}F^{ij}\partial_i\partial_j
    \left[
        \frac{1}{2T^2}(J_k\theta^k)^2
    \right] 
    &=
    \frac{1}{4T}F^{ij}\partial_i\partial_j
    \left(
        J_k\theta^k J_\ell\theta^\ell
    \right) \\
    &=
    \frac{1}{4T}F^{ij}\partial_i
    \left(
        J_k\delta_j^k J_\ell\theta^\ell
        +
        J_k\theta^k J_\ell\delta_j^\ell
    \right) \\
    &=
    \frac{1}{4T}F^{ij}\partial_i
    \left(
        J_jJ_\ell\theta^\ell
        +
        J_k\theta^kJ_j
    \right) \\
    &=
    \frac{1}{4T}F^{ij}
    \left(
        J_jJ_\ell\delta_i^\ell
        +
        J_k\delta_i^kJ_j
    \right) \\
    &=
    \frac{1}{4T}F^{ij}
    \left(
        J_jJ_i+J_iJ_j
    \right) \\
    &=
    \frac{1}{2T}F^{ij}J_iJ_j .
\end{align}
This notation is used in Special and General Relativity, thus its name.  For a
longer and more rigorous introduction, see for instance Sec.\@ 1.2.2 of the
lecture notes by \citet{matthias_theoastro}.

\section{Gaussian Measure Identity}
\label{app:gaussian_measure_identity}

Here we briefly motivate the identity 
\begin{equation} \label{eq:qft_id}
    \langle f(\theta)\rangle
    =
    \left.
    \exp\left(
        \frac{1}{2}C^{ij}\partial_i\partial_j
    \right)
    f(\theta)
    \right|_{\theta=0}.
\end{equation}
used in Eq.~\eqref{eq:qft_expr}. Let \(\theta\sim\mathcal{N}(0,C)\) and expand
$f$ around the origin,
\begin{align}
    f(\theta)
    =
    f(0)
    +
    \partial_i f(0)\theta^i
    +
    \frac{1}{2}\partial_i\partial_j f(0)\theta^i\theta^j
    +
    \frac{1}{3!}\partial_i\partial_j\partial_k f(0)\theta^i\theta^j\theta^k
    +
    \frac{1}{4!}\partial_i\partial_j\partial_k\partial_\ell f(0)
    \theta^i\theta^j\theta^k\theta^\ell
    +\mathcal{O}(\theta^5) .
\end{align}
Since the Gaussian is centered, all odd moments vanish. The first non-vanishing
moments are
\begin{equation}
    \langle \theta^i\theta^j\rangle = C^{ij},
    \qquad
    \langle \theta^i\theta^j\theta^k\theta^\ell\rangle
    =
    C^{ij}C^{k\ell}
    +
    C^{ik}C^{j\ell}
    +
    C^{i\ell}C^{jk}.
\end{equation}
Thus the left-hand side of \eqref{eq:qft_id} becomes
\begin{align}
    \langle f(\theta)\rangle
    =
    f(0)
    +
    \frac{1}{2}C^{ij}\partial_i\partial_j f(0)
    +
    \frac{1}{4!}
    \left(
        C^{ij}C^{k\ell}
        +
        C^{ik}C^{j\ell}
        +
        C^{i\ell}C^{jk}
    \right)
    \partial_i\partial_j\partial_k\partial_\ell f(0)
    +\mathcal{O}(C^3) .
    \label{eq:gaussian_moment_expansion}
\end{align}
On the other hand, expanding the differential operator gives
\begin{align}
    \left.
    \exp\left(
        \frac{1}{2}C^{ij}\partial_i\partial_j
    \right)
    f(\theta)
    \right|_{\theta=0}
    &=
    \left.
    \left[
        1
        +
        \frac{1}{2}C^{ij}\partial_i\partial_j
        +
        \frac{1}{2!}
        \left(
            \frac{1}{2}C^{ij}\partial_i\partial_j
        \right)^2
        +\mathcal{O}(C^3)
    \right]
    f(\theta)
    \right|_{\theta=0}
    \nonumber \\
    &=
    f(0)
    +
    \frac{1}{2}C^{ij}\partial_i\partial_j f(0)
    +
    \frac{1}{8}C^{ij}C^{k\ell}
    \partial_i\partial_j\partial_k\partial_\ell f(0)
    +\mathcal{O}(C^3) .
    \label{eq:gaussian_operator_expansion}
\end{align}
which is equal to the left-hand side by comparison of coefficients. Identities
of this kind are heavily used for example in quantum field theory and discussed
in more detail by, e.\,g.\@, \citet{zinnjustin}. Incidentally, one can also use this
identity to compute moments of a distribution by taking derivatives of a normalising
flow that approximates it \citep{partition_flow}.

\section{Toy-Model Thermodynamics}
\label{app:toy_model_thermodynamics}

Here we collect the short derivations of the thermodynamic theories used in
Sections~\ref{sec:gauss} and~\ref{sec:rosenbrock}. Throughout, $\pi \equiv 1$
and constants in the likelihood that do not depend on $T$ or $J$ are omitted.

\paragraph{Gaussian likelihood.}
Starting with Eq.~\eqref{eq:gaussian_likelihood}, the canonical partition sum is
computed by completing the square.
\begin{align}
Z(T,J)
&=
\int \dd^n\theta\,
\exp\left[
-\frac{1}{2T}F_{ij}\theta^i\theta^j
+\frac{1}{T}J_i\theta^i
\right] \nonumber \\
&=
\exp\left(\frac{1}{2T}F^{ij}J_iJ_j\right)
\int \dd^n\theta\,
\exp\left[
-\frac{1}{2T}F_{ij}
\left(\theta^i-F^{ik}J_k\right)
\left(\theta^j-F^{jl}J_l\right)
\right] \nonumber \\
&=
\sqrt{\frac{(2\pi T)^n}{\det F}}\,
\exp\left(\frac{1}{2T}F^{ij}J_iJ_j\right).
\end{align}
Thus
\begin{align}
F(T,J)
&=
-T\log Z
=
-\frac{Tn}{2}\log(2\pi T)
+\frac{T}{2}\log\det F
-\frac{1}{2}F^{ij}J_iJ_j, \nonumber \\
S(T,J)
&=
-\pdv{F}{T}
=
\frac{n}{2}\log(2\pi T)
+\frac{n}{2}
-\frac{1}{2}\log\det F .
\end{align}
The entropy is independent of $J$, hence
\begin{equation}
C=T\pdv{S}{T}=\frac{n}{2},
\qquad
h^i=T\pdv{S}{J_i}=0,
\qquad
\dot Q=C\dot T+h^i\dot J_i=\frac{n}{2}\dot T .
\end{equation}

\paragraph{Rosenbrock likelihood.}
For Eq.~\eqref{eq:rosenbrock_def}, set $x=\theta^0$ and
$u=\theta^1-c^2x^2$. Then $\dd\theta^0\dd\theta^1=\dd x\dd u$ and
again we complete the squares.
\begin{align}
Z(T,J)
&=
\int \dd x\,\dd u\,
\exp\left[
-\frac{1}{2T}\left(a^2x^2+b^2u^2\right)
+\frac{1}{T}\left(J_0x+J_1u+c^2J_1x^2\right)
\right] \nonumber \\
&=
\int \dd x\,
\exp\left[
-\frac{a^2-2c^2J_1}{2T}x^2
+\frac{J_0}{T}x
\right]
\int \dd u\,
\exp\left[
-\frac{b^2}{2T}u^2
+\frac{J_1}{T}u
\right] \nonumber \\
&=
\exp\left[
\frac{1}{2T}
\frac{J_0^2}{a^2-2c^2J_1}
\right]
\int \dd x\,
\exp\left[
-\frac{a^2-2c^2J_1}{2T}
\left(
x-\frac{J_0}{a^2-2c^2J_1}
\right)^2
\right] \nonumber \\
&\quad \times
\exp\left[
\frac{1}{2T}
\frac{J_1^2}{b^2}
\right]
\int \dd u\,
\exp\left[
-\frac{b^2}{2T}
\left(
u-\frac{J_1}{b^2}
\right)^2
\right] \nonumber \\
&=
\frac{2\pi T}{b\sqrt{a^2-2c^2J_1}}
\exp\left[
\frac{1}{2T}
\left(
\frac{J_1^2}{b^2}
+
\frac{J_0^2}{a^2-2c^2J_1}
\right)
\right],
\end{align}
valid for $a^2-2c^2J_1>0$. The free energy is
\begin{equation}
F(T,J)
=
-T\log(2\pi T)
+T\log b
+\frac{T}{2}\log(a^2-2c^2J_1)
-\frac{1}{2}
\left(
\frac{J_1^2}{b^2}
+
\frac{J_0^2}{a^2-2c^2J_1}
\right),
\end{equation}
and therefore
\begin{equation}
S(T,J)
=
-\pdv{F}{T}
=
1+\log(2\pi T)-\log b
-\frac{1}{2}\log(a^2-2c^2J_1).
\end{equation}
Taking the derivatives that enter Eq.~\eqref{eq:heat_rate} gives
\begin{equation}
C=T\pdv{S}{T}=1,
\qquad
h^0=T\pdv{S}{J_0}=0,
\qquad
h^1=T\pdv{S}{J_1}
=
\frac{Tc^2}{a^2-2c^2J_1},
\end{equation}
so that the heat rate is
\begin{equation}
\dot Q
=
\dot T
+
\frac{Tc^2}{a^2-2c^2J_1}\dot J_1 .
\end{equation}

\section{Experimental Hyperparameters}
\label{app:hyperparameters}

Table~\ref{tab:hyperparameters} summarises the hyperparameters used for the
numerical experiments presented in this paper.
Thinning refers to subsampling the MCMC chain during sampling to reduce
autocorrelation, while skip steps refers to downsampling the trajectory for
plotting and analysis to reduce visual clutter and file size.
The Rosenbrock $c$ sweep (Figure \ref{fig:rosenbrock_net_work_vs_c}) uses the
hyperparameters in the nonzero-work column, except that $c$ is varied.

\begin{table}[htbp]
\centering
\caption{Hyperparameters for the numerical experiments.}
\label{tab:hyperparameters}
\resizebox{\textwidth}{!}{\begin{tabular}{@{}lcccccc@{}}
\toprule
& Gaussian (across) & Gaussian (along) & Rosenbrock (nonzero) & Rosenbrock (zero) & SN Ia (surrogate) & SN Ia (full) \\
\midrule
Likelihood & \eqref{eq:gaussian_likelihood} & \eqref{eq:gaussian_likelihood} & \eqref{eq:rosenbrock_def} & \eqref{eq:rosenbrock_def} & \eqref{eq:sn1a_surrogate} & \eqref{eq:sn1a_likelihood} \\
Params & $a=1$, $b=10$ & $a=1$, $b=10$ & $a=2$, $b=10$, $c=0.5$ & $a=2$, $b=10$, $c=0.5$ & $M=-19.25, H_0=73.28$ & $M=-19.25, H_0=73.28$ \\
$T_c$ & 1.0 & 1.0 & 1.0 & 1.0 & 1.0 & 1.0 \\
$T_h$ & 2.0 & 2.0 & 3.0 & 3.0 & 10.0 & 10.0 \\
$J_{(1)}$ & $(5.5, 0)$ & $(0, 5.5)$ & $(0, -4)$ & $(6, 0)$ & $(0.1, 1.0)$ & $(0.1, 1.0)$ \\
$J_{(2)}$ & $(5, 0)$ & $(0, 5)$ & $(0, -6)$ & $(5, 0)$ & $(0.15, 1.5)$ & $(0.15, 1.5)$ \\
$\epsilon$ & $10^{-3}$ & $10^{-3}$ & $10^{-3}$ & $5\times10^{-4}$ & $4\times10^{-4}$ & $4\times10^{-4}$ \\
Sampler & RMH & RMH & HMC & HMC & HMC & HMC \\
Int.\ steps & --- & --- & 10 & 10 & 10 & 10 \\
Chains & $10^2$ & $10^2$ & $10^2$ & $10^2$ & $10^2$ & $10^2$ \\
Steps & $10^5$ & $10^5$ & $3\times10^5$ & $3\times10^5$ & $10^5$ & $10^5$ \\
Burn-in & $10^3$ & $10^3$ & $10^4$ & $10^4$ & $10^4$ & $10^4$ \\
Thin & 10 & 10 & $10^2$ & $10^2$ & $10^2$ & $10^2$ \\
Step size & 0.5 & 0.5 & 0.1 & 0.1 & 0.05 & 0.05 \\
Skip & $10^2$ & $10^2$ & $10^2$ & $10^2$ & $10^2$ & $10^2$ \\
Figures & \ref{fig:anisotropic_gaussian_energy}, \ref{fig:anisotropic_gaussian_heat_work}, \ref{fig:gaussian_q_net_vs_n_cycles}  & \ref{fig:anisotropic_gaussian_J_along} & \ref{fig:rosenbrock_summary}, \ref{fig:rosenbrock_net_work_vs_c} & \ref{fig:rosenbrock_summary} & \ref{fig:sn1a} & \ref{fig:sn1a} \\
\bottomrule
\end{tabular}}
\end{table}
\end{document}